\documentclass[]{article}

\usepackage{times}

\usepackage{ijcai16}
\usepackage{url,tabularx,array,enumerate,amsmath}
\usepackage{graphicx}
\usepackage[english]{babel}
\usepackage{amsthm}
\usepackage{enumerate}
\usepackage{algorithm}
\usepackage{algorithmicx}
\usepackage{algpseudocode}
\usepackage{caption}
\usepackage{subcaption}
\usepackage[multiple]{footmisc}

\usepackage{url}
\newtheorem{mydef}{Definition}
\newtheorem{theorem}{Theorem}

\newtheorem{lemma}{Lemma}

\newcommand{\indep}{\rotatebox[origin=c]{90}{$\models$}}
\newcommand{\notindep}{\not\!\perp\!\!\!\perp}

\title{Incorporating Knowledge into Structural Equation Models \\using Auxiliary Variables}

\author{Bryant Chen\\
UCLA\\
\url{bryantc@cs.ucla.edu}
\And
Judea Pearl\\
UCLA\\
\url{judea@cs.ucla.edu}
\And
Elias Bareinboim\\
Purdue University\\
\url{eb@purdue.edu}}

\begin{document}

\maketitle

\begin{abstract}
In this paper, we extend graph-based identification methods by allowing background knowledge in the form of non-zero parameter values. Such information could be obtained, for example, from a previously conducted randomized experiment, from substantive understanding of the domain, or even an identification technique. To incorporate such information systematically, we propose the addition of auxiliary variables to the model, which are constructed so that certain paths will be conveniently cancelled.  This cancellation allows the auxiliary variables to help conventional methods of identification (e.g., single-door criterion, instrumental variables, half-trek criterion), as well as model testing (e.g., d-separation, over-identification).  Moreover, by iteratively alternating steps of identification and adding auxiliary variables, we can improve the power of existing identification methods via a bootstrapping approach that does not require external knowledge. We operationalize this method for simple instrumental sets (a generalization of instrumental variables) and show that the resulting method is able to identify at least as many models as the most general identification method for linear systems known to date. We further discuss the application of auxiliary variables to the tasks of model testing and z-identification.
\end{abstract}

\section{Introduction}

Many researchers, particularly in economics, psychology, epidemiology, and the social sciences, use linear structural equation models (SEMs) to describe the causal and statistical relationships between a set of variables, predict the effects of interventions and policies, and to estimate parameters of interest \cite{bollen:pea393}.  A linear SEM consists of a set of equations of the form, 
\begin{equation*}
X = \Lambda X+ U,
\end{equation*}
where $X = [x_1 , ... , x_n]^t$ is a vector containing the model variables, $\Lambda$ is a matrix containing the \emph{coefficients} of the model, which convey the strength of the causal relationships, and $U = [u_1 , ..., u_n]^{t}$ is a vector of error terms, which represents omitted or latent variables and are assumed to be normally distributed.  The matrix $\Lambda$ contains zeroes on the diagonal, and $\Lambda_{ij} = 0$ whenever $x_i$ is not a cause of $x_j$.  The covariance matrix of $X$ will be denoted by $\Sigma$ and the covariance matrix over the error terms, U, by $\Omega$.  The entries of $\Lambda$ and $\Omega$ are the \emph{model parameters}.  In this paper, we will restrict our attention to \emph{semi-Markovian} models \cite{pearl:2k}, models where the rows of $\Lambda$ can be arranged so that it is lower triangular.  

When the coefficients are known, then total effects, direct effects, and counterfactuals can be computed from them directly \cite{pearl:09,chen:pea14}.  However, in order to be able to compute these coefficients, we must utilize domain knowledge in the form of exclusion and independence restrictions \cite[p. 704]{pearl:95}.  Exclusion restrictions represent assumptions that a given variable is not a direct cause of another, while independence restrictions represent assumptions that no latent confounders exists between two variables.  Algebraically, these assumptions translate into restrictions on entries in the coefficient matrix, $\Lambda$, and error term covariance matrix, $\Omega$, to zero.  

Determining whether model parameters can be expressed in terms of the probability distribution, which is necessary to be able to estimate them from data, is the problem of \emph{identification}.  When it is not possible to uniquely express the value of a model parameter in terms of the probability distribution, we will say that the parameter is \emph{not identifiable}.\footnote{We will also use the term ``identifiable" with respect to the model as a whole.  When the model contains an unidentified coefficient, the model is not identified.} In linear systems, this generally takes the form of expressing a parameter in terms of the covariance matrix over the observable variables.  

To our knowledge, the most general method for determining model identification is the half-trek criterion \cite{foygel:12}.
Identifying individual structural coefficients can be accomplished using the single-door criterion (i.e. identification using regression) \cite{pearl:09,chen:pea14}, instrumental variables \cite{wright:25,wright:28} (see \cite{brito:pea02a}, \cite{pearl:09}, or \cite{chen:pea14} for a graphical characterization), instrumental sets \cite{brito:pea02a}, and the general half-trek criterion \cite{chen:15}, which generalizes the half-trek criterion for individual coefficients rather than entire models.  Finally, d-separation \cite{pearl:09} and overidentification \cite{pearl:04,chen:etal14} provide the means to enumerate testable implications of the model, which can be used to test it against data.  


Each of these methods only utilize restrictions on the entries of $\Lambda$ and $\Omega$ to zero.  In this paper, we introduce \emph{auxiliary variables}, which can be used to incorporate knowledge of non-zero coefficient values into existing methods of identification and model testing.  The intuition behind auxiliary variables is simple: if the coefficient from variable $w$ to $z$ is known, then we would like to remove the direct effect of $w$ on $z$ by subtracting it from $z$.  We do this by creating a variable $z^* = z-\alpha w$ and using it as a proxy for $z$.  In some cases, $z^*$ may allow the identification of parameters or testable implications using the aforementioned methods when $z$ could not.

While intuitively simple, auxiliary variables are able to greatly increase the power of existing identification methods, even without external knowledge of coefficient values.  We propose a bootstrapping procedure whereby coefficients are iteratively identified using simple instrumental sets and then used to generate auxiliary variables, which enable the identification of previously unidentifiable coefficients.  We prove that this method enhances the instrumental set method to the extent that it is able to subsume the relatively more complex general half-trek criterion (henceforth, g-HTC).  

The notion of ``subtracting out a direct effect'' in order to turn a variable into an instrument was first noted by \cite{shardell:15} when attemping to identify the total effect of $x$ on $y$. It was noticed that in certain cases, the violation of the independence restriction of a potential instrument $z$ (i.e. $z$ is not independent of the error term of $y$) could be remedied by identifying, using ordinary least squares regression, and then subtracting out the necessary direct effects on $y$.  In this paper, we generalize and operationalize this notion so that it can be used on arbitrary sets of known coefficient values and be utilized in conjunction with graphical methods for identification and enumeration of testable implications. 

The paper is organized as follows: Sec. \ref{sec:prelim} reviews notation and graphical notions that will be used in the paper.  In sec. \ref{sec:EC}, we introduce and formalize auxiliary variables and auxiliary instrumental sets.  Additionally, we give a sufficient graphical condition for the identification of a set of coefficients using auxiliary instrumental sets.  In sec. \ref{sec:IS}, we show that auxiliary instrumental sets subsume the g-HTC.  Finally, in sec. \ref{sec:discuss}, we discuss additional applications of auxiliary variables, including identifying testable implications and z-identification \cite{bareinboim:pea12}. 

\section{Preliminaries}
\label{sec:prelim}

The causal graph or path diagram of a SEM is a graph, $G=(V,D,B)$, where $V$ are nodes or vertices, $D$ directed edges, and $B$ bidirected edges.  The nodes represent model variables.  Directed eges encode the direction of causality, and for each coefficient $\Lambda_{ij}\neq 0$, an edge is drawn from $x_i$ to $x_j$.  Each directed edge, therefore, is associated with a coefficient in the SEM, which we will often refer to as its structural coefficient.  The error terms, $u_i$, are not shown explicitly in the graph.  However, a bidirected edge between two nodes indicates that their corresponding error terms may be statistically dependent while the lack of a bidirected edge indicates that the error terms are independent.  

If a directed edge, called $(x, y)$, exists from $x$ to $y$ then $x$ is a parent of $y$.  The set of parents of $y$ is denoted $Pa(y)$.  Additionally, we call $y$ the head of $(x,y)$ and $x$ the tail.  The set of tails for a set of directed edges, $E$, is denoted $Ta(E)$ while the set of heads is denoted $He(E)$.  For a node, $v$, the set of edges for which $He(E)=v$ is denoted $Inc(v)$. Finally, the set of nodes connected to $y$ by a bidirected arc are called the siblings of $y$ or $Sib(y)$. 

A \emph{path} from $x$ to $y$ is a sequence of edges connecting the two nodes.  A path may go either along or against the direction of the edges.  A non-endpoint node $w$ on a path is said to be a \emph{collider} if the edges preceding and following $w$ both point to $w$.  

A path between $x$ and $y$ is said to be \emph{unblocked given a set $Z$}, with $x, y\notin Z$ if every noncollider on the path is not in $Z$ and every collider on the path is in $An(Z)$ \cite{pearl:09}, where $An(Z)$ are the ancestors of $Z$.  Unblocked paths of the form $a\rightarrow ... \rightarrow b$ or $a \leftarrow ... \leftarrow b$ are \emph{directed} paths.  Any unblocked path that is not a directed path is a \emph{divergent} path.

$\sigma (x, y)$ denotes the covariance between two random variables, $x$ and $y$, and $\sigma_M (x,y)$ is the covariance between random variables $x$ and $y$ induced by the model $M$.  $(x\indep y)$ denotes that $x$ is independent of $y$, and similarly, $(x\indep y)_M$ denotes that $x$ is independent of $y$ according to the model, $M$.  We will assume without loss of generality that the model variables have been standardized to mean 0 and variance 1.

We will also utilize a number of definitions around half-treks \cite{foygel:12}.

\begin{figure*}
\centering
\begin{subfigure}[t]{.21\textwidth}
\caption{}
\label{fig:IS+Ra}
\includegraphics[width=\textwidth]{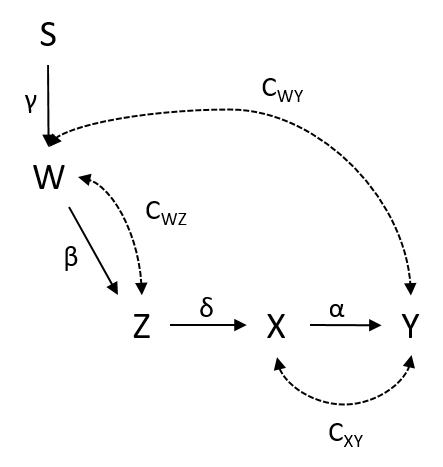}
\end{subfigure}
\begin{subfigure}[t]{.24\textwidth}
\caption{}
\label{fig:IS+Rb}
\includegraphics[width=\textwidth]{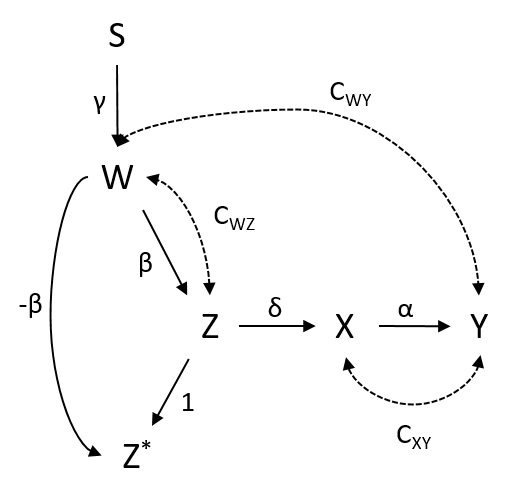}
\end{subfigure}
\caption{(a) $\alpha$ is not identified using instruments (b) The $\beta$-augmented graph, where $\alpha$ is identified using $z^*$ as an instrument}
\end{figure*}

\begin{figure}
\centering
\begin{subfigure}[t]{.128\textwidth}
\caption{}
\label{fig:de_aux1}
\includegraphics[width=\textwidth]{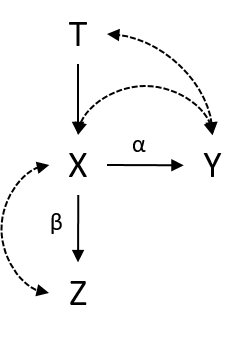}
\end{subfigure}
\begin{subfigure}[t]{.105\textwidth}
\caption{}
\label{fig:de_aux2}
\includegraphics[width=\textwidth]{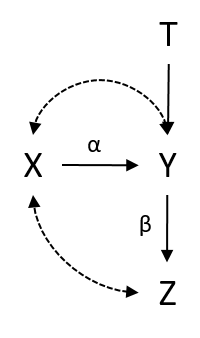}
\end{subfigure}
\caption{$\alpha$ is identified in (a) and (b) using $z^*$ as an auxiliary instrument after identification of $\beta$ using $t$ as an instrument}
\end{figure}

\begin{mydef}
A \emph{half-trek}, $\pi$, from $x$ to $y$ is an unblocked path from $x$ to $y$ that either begins with a bidirected arc and then continues with directed edges towards $y$ or is simply a directed path from $x$ to $y$. 
\end{mydef}
\noindent We will denote the set of nodes connected to a node, $v$, via half-treks $htr(v)$.  For example, in Figure \ref{fig:IS-ECa}, $w \leftrightarrow z \rightarrow x \rightarrow y$ and $w\rightarrow z\rightarrow x\rightarrow y$ are both half-treks from $w$ to $y$.  However, $z^* \leftarrow w \rightarrow z$ in Figure \ref{fig:IS-ECb} is not a half-trek from $z^*$ to $z$ because it begins with an arrow pointing to $z^*$.

\begin{mydef}
For a given path, $\pi$, from $x$ to $y$, Left($\pi$) is the set of nodes, if any, that has a directed edge leaving it in the direction of $x$ in addition to $x$.  Right($\pi$) is the set of nodes, if any, that has a directed edge leaving it in the direction of $y$ in addition to $y$.  
\end{mydef}

For example, consider the path $\pi = x\leftarrow v^L_1 \leftarrow ... \leftarrow v^L_k \leftarrow v^T \rightarrow v^R_j \rightarrow ... \rightarrow v^R_1 \rightarrow y$.  In this case, Left($\pi$) $= \cup_{i=1}^{k} v^L_i \cup \{x, v^T\}$ and Right($\pi$) $= \cup_{i=1}^{j} v^R_i \cup \{y, v^T\}$.  $v^T$ is a member of both Right$(\pi)$ and Left($\pi)$.  

\begin{mydef}
A set of paths, ${\pi_1, ..., \pi_n}$, has \emph{no sided intersection} if for all $\pi_i , \pi_j \in \{\pi_1, ..., \pi_n\}$ such that $\pi_i \neq \pi_j$, Left$(\pi_i)\cap$Left$(\pi_j)$=Right$(\pi_i)\cap$Right$(\pi_j)=\emptyset$.
\end{mydef}

Consider the set of paths $\{\pi_1 = x\rightarrow y, \pi_2 = z\leftrightarrow x \rightarrow w\}$. This set has no sided intersection, even though both paths contain $x$, because Left($\pi_1$) $= \{x\}$, Left($\pi_2$) $= \{z\}$, Right($\pi_1$) $= \{y\}$, and Right($\pi_2$) $= \{x,w\}$.  In contrast, $\{\pi_1 = x\rightarrow y, \pi_2 = z\rightarrow x\rightarrow w\}$ does have a sided intersection because $x$ is in both Right($\pi_1$) and Right($\pi_2$).

Wright's rules \cite{wright:21} allows us to equate the model-implied covariance, $\sigma_M (x,y)$, between any pair of variables, $x$ and $y$, to the sum of products of parameters along unblocked paths between $x$ and $y$.\footnote{Wright's rules characterize the relationship between the covariance matrix and model parameters.  Therefore, any question about identification using the covariance matrix can be decided by studying the solutions for this system of equations.  However, since these equations are polynomials and not linear, it can be very difficult to analyze identification of models using Wright's rules \cite{brito:04}.}  Let $\Pi= \{\pi_1 , \pi_2 ,... , \pi_k\}$ denote the unblocked paths between $x$ and $y$, and let $p_i$ be the product of structural coefficients along path $\pi_i$.  Then the covariance between variables $x$ and $y$ is $\sum_i p_i$.
We will denote the expression that Wright's rules gives for $\sigma(x,y)$ in graph $G$, $W_G (x,y)$.  

Instrumental variables (IVs) is one of the most common methods of identifying parameters in linear models.  The ability to use an instrumental set to identify a set of parameters when none of those parameters are identifiable individually using IVs was first proposed by \cite{brito:pea02a}.  

\begin{figure*}
\centering
\begin{subfigure}[t]{.32\textwidth}
\caption{}
\label{fig:IS-ECa}
\includegraphics[width=\textwidth]{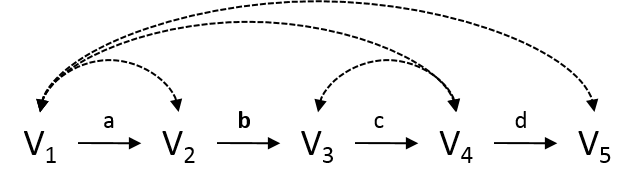}
\end{subfigure}
\begin{subfigure}[t]{.32\textwidth}
\caption{}
\label{fig:IS-ECb}
\includegraphics[width=\textwidth]{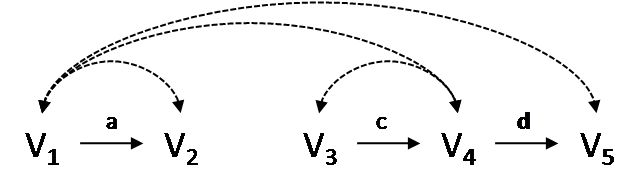}
\end{subfigure}
\begin{subfigure}[t]{.32\textwidth}
\caption{}
\label{fig:IS-ECc}
\includegraphics[width=\textwidth]{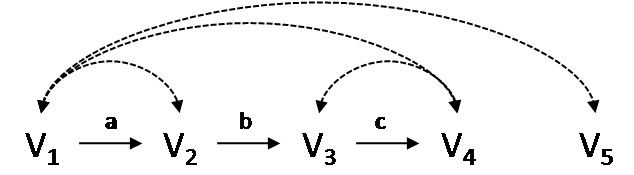}
\end{subfigure}
\caption{(a) $b$ is identified using either $v_2$ or $v_1$ as an instrument (b) $d$ is identified using $v_3^*$ as an auxiliary instrument (c) $a$ and $c$ are identified using $v_5^*$ as an auxiliary instrument}
\end{figure*}

\begin{mydef}[Simple Instrumental Set]
\label{def:IVset}
$Z$ is a \emph{simple instrumental set} for the coefficients associated with edges $E = \{x_1 \rightarrow y, ..., x_k \rightarrow y\}$ if the following conditions are satisfied.
\begin{enumerate}[(i)]
\item $|Z| = k$.
\item Let $G_E$ be the graph obtained from $G$ by deleting edges $x_1 \rightarrow y, ..., x_k \rightarrow y$. Then, $(z_i\indep y)_{G_E}$ for all $i\in \{1, 2, ..., k\}$.\footnote{This condition can also be satisfied by conditioning on a set of covariates without changing the results below, but for simplicity we will not consider this case. When conditioning on a set of covariates, $Z$ is called a \emph{generalized instrumental set}.}
\item There exist unblocked paths $\Pi = \{\pi_1 , \pi_2 ,... , \pi_k\}$ such that $\pi_i$  is an unblocked path from $z_i$ to $x_i$ and $\Pi$ has no sided intersection.
\end{enumerate}
\end{mydef}

If $Z$ is a simple instrumental set for $E$, then we can use Wright's rules to obtain a set of $|k|$ linearly independent equations in terms of the coefficients, enabling us to solve for the coefficients \cite{brito:pea02a}.  

\section{Auxiliary Variables}
\label{sec:EC}

We start this section by motivating auxiliary variables through an example. Consider the structural system depicted in Figure \ref{fig:IS+Ra}.  In this system, the structural coefficient $\alpha$ is not identifiable using instrumental variables or instrumental sets. To witness, note that $z$, $w$, and $s$ all fail to qualify as instruments due to the spurious paths, $z\leftarrow w \leftrightarrow y$, $w\leftrightarrow y$, and $s\leftrightarrow y$, respectively\footnote{Note that even if we consider \emph{conditional instruments} \cite{brito:pea02a}, these paths cannot be blocked, and identification is not possible.}.  If the coefficient $\beta$ is known,\footnote{The coefficient $\beta$ may be available through different means, for instance, from a smaller randomized experiment, pilot study, or substantive knowledge, just to cite a few. In this specific case, however, $\beta$ can be identified directly without invoking external information by simply using $S$ as an instrument.} we can add an \emph{auxiliary variable}, $z^* = z-\beta w$, to the model.  Subtracting $\beta w$ from $z$ cancels the effect of $w$ on $z$ so that $w$ has no effect on $z^* = (\beta w + u_w) - \beta w = u_w$.  Now,  $z^*$ is an instrument for $\alpha$.  The sum of products of parameters along back-door paths from $z^*$ to $y$ is equal to 0 and $\alpha = \frac{\sigma(z^*,y)}{\sigma(z^*,x)}$.  

Surprisingly, auxiliary variables can even be used to generate instruments from effects of $x$ and $y$.  For example, consider Figures \ref{fig:de_aux1} and \ref{fig:de_aux2}.  In both examples, $z$ is clearly not an instrument for $\alpha$.  However, in both cases, $\beta$ is identifiable using $t$ as an instrument, allowing us to construct the auxiliary variable, $z^* = z-\alpha x$, which does qualify as an instrument for $\alpha$ (see Theorem \ref{thm:aux-IS} below).  

The following definition establishes the $\beta$-\emph{augmented} model, which incorporates the $z^*$ variable into the model.\footnote{\cite{chan:kuroki10} also gave a graphical criterion for identification of a coefficient using descendants of $x$.  $\alpha$ in Figure \ref{fig:de_aux1} can also be identified using their method.}

\begin{mydef}
\label{def:aug}
Let $M$ be a structural causal model with associated graph $G$ and a set of directed edges $E$ such that their coefficient values are known. The \emph{$E$-augmented} model, $M^{E+}$, includes all variables and structural equations of $M$ in addition to new auxiliary variables, $y^*_1 , ... y^*_k$, one for each variable in $He(E)=\{y_1 , ... , y_k\}$  such that the structural equation for $y^*_i$ is $y^*_i = y_i - \Lambda_{X_i y_i} X^t_i$, where $X_i=Ta(E)\cap Pa(y_i)$, for all $i\in \{1 ,... , k\}$.  The corresponding graph is denoted $G^{E+}$.  

\end{mydef}


For example, let $M$ and $G$ be the model and graph depicted in Figure \ref{fig:IS+Ra}.  The $\beta$-augmented model is obtained by adding a new variable $z^* = z - \beta w$ to $M$.  The corresponding graph, $G^{\beta+}$, is shown in Figure \ref{fig:IS+Rb}.  The following lemma establishes that the covariance between any two variables in $V^* = V \cup He^*(E)$ can be obtained using Wright's rules on $G^{E+}$, where $V$ is the set of variables in $M$ and $He^*(E)$ is the set of variables added to the augmented model.\footnote{Note that auxiliary variables may not have a variance of 1.  We will see that this does not affect the results of the paper since the covariance between model variables implied by the graph is correct, even after the addition of auxiliary variables.}

\begin{lemma}
\label{lem:aug}
Given a linear structural model, $M$, with induced graph $G$, and a set of directed edges $E$ with known coefficient values, $\sigma (w,v) = W_{G^{E+}}(w,v)$, where $w,v \in V^*$ and $w\neq v$.\footnote{See Appendix for proofs of all lemmas.}
\end{lemma}

The above lemma guarantees that the covariance between variables implied by the augmented graph is correct, and Wright's rules can be used to identify coefficients in the model $M$.  For example, using Wright's rules on $G^{\beta+}$, depicted in Figure \ref{fig:IS+Rb}, yields
\begin{align*} \sigma(z^* , y) =& (1*\beta - \beta) (C_{WY} + \gamma C_{SY}) \\ &+(1-\beta^2-\beta C_{WZ})\delta \alpha \\ 
=& \alpha(\delta- \beta^2*\delta - \beta*C_{wz} *\delta) \end{align*} 
and 
\begin{align*}\sigma(z^*, x) =& \delta -\beta^2*\delta - \beta*C_{wz} *\delta
\end{align*} so that $\alpha = \frac{\sigma(z^* , y )}{\sigma(z^* , x)}$.  As a result, $z^*$ can be used as an instrumental variable for $\alpha$ when $z$ clearly could not.  

\begin{mydef}[Auxiliary Instrumental Set]
\label{def:auxIS}
Given a semi-Markovian linear SEM with graph $G$ and a set of directed edges $E_Z$ whose coefficient values are known, we will say that a set of nodes, $Z$, in $G$ is an \emph{auxiliary instrumental set} or \emph{aux-IS} for $E$ if $Z^* = (Z\setminus A)\cup A^*$ is an instrumental set for $E$ in $G^{E_Z +}$, where $A$ is the set of variables in $Z$ that have auxiliary variables in $G^{E_Z+}$. 
\end{mydef}

The following lemma characterizes when an auxiliary variable will be independent of a model variable and is used to prove Theorem \ref{thm:aux-IS}.

\begin{lemma}
\label{lem:sep}
Given a semi-Markovian linear SEM with graph $G$, $(z^*\indep y)_{G^{E_{z}+}}$ if and only if $z$ is d-separated from $y$ in $G_{E_z}$, where $E_z\subseteq Inc(z)$ and $G_{E_z}$ is the graph obtained when $E_z$ is removed from $G$. 
\end{lemma}
%

The following theorem provides a simple method for recognizing auxiliary instrumental sets using the graph, $G$.

\begin{theorem}
\label{thm:aux-IS}
Let $E_Z$ be a set of directed edges whose coefficient values are known.  A set of directed edges, $E=\{(x_1, y) , ..., (x_k, y)\}$, in a graph, $G$, is identified if there exists $Z$ such that:
\begin{enumerate}
\item $|Z|=k$,
\item for all $z_i \in Z$, $(z_i \indep y)_{G_{E\cup E_{z_i}}}$, where $E_{z_i}=E_Z\cap Inc(z_i)$ and $G_{E\cup E_{z_i}}$ is the graph obtained by removing the edges in $E\cup E_{z_i}$ from $G$, and
\item there exists unblocked paths $\Pi = \{\pi_1 , \pi_2 ,... , \pi_k\}$ such that $\pi_i$ is a path from $z_i$ to $x_i$ and $\Pi$ has no sided intersection.
\end{enumerate}
\end{theorem}
If the above conditions are satisfied then $Z$ is an auxiliary instrumental set for $E$.
\begin{proof}
We will show that $Z^*$ is an instrumental set in $G^{E_Z+}$.  First, note that if $E_{Z} = \emptyset$, then $Z$ is an instrumental set in $G$ and we are done.  We now consider the case when $E_{Z} \neq \emptyset$.  Since $|Z^*|=|Z|-|A|+|A*| = |Z|-|A|+|A|=|Z|$, $|Z^*|=|E|$, IS-(i) is satisfied.  Now, we show that IS-(iii) is satisfied.  For each $z_i \in Z$, let $\pi_{z_i} \in \Pi$ be the path in $\Pi$ from $z_i$ to $Ta(E)$.  Now, for each $a^*_i\in A^*$, let $\pi_{a^*_i}$ be the concatenation of path $a^* \leftarrow a$ with $\pi_{a_i}$.  It should be clear that $\Pi\setminus \{\pi_{a_i}\} \cup \{\pi_{a^*_i}\}$ satisfies IS-(iii) in $G^{E_{Zy}+}$.  Lastly, we need to show that IS-(ii) is also satisfied.   

First, if $z_i \in Z\setminus A$, then $(z_i \indep y)_{G_E}$.  It follows that $(z_i \indep y)_{G^{E_Z+}_E}$ since no new paths from $z_i$ to $y$ can be generated by adding the auxiliary nodes (see Lemma 8 in Appendix).  Now, we know that $(a^*_i \indep y)_{G^{E_{a_i}+}_E}$ from (ii) and Lemma \ref{lem:sep}.  Finally, since adding auxiliary variables cannot generate new paths between the existing nodes, we know that $(a^*_i \indep y)_{G^{E_Z+}}$, and we are done.  

$(a^*_i\indep y)_{G^{E_{Z}+}_E}$ for all $a_i \in A$ follows from (ii), Lemma \ref{lem:sep}, and the fact that no new paths from $a_i$ to $y$ can be generated by adding auxiliary nodes, proving the theorem.
\end{proof}

To see how Theorem \ref{thm:aux-IS} can be used to identify auxiliary instrumental sets, consider Figure \ref{fig:IS-ECa}.  Using instrumental sets, we are able to identify $b$, but no other coefficients.  Once $b$ is identified, $d$ can be identified using $v_3^*$ as an instrument in $G^{b+}$ since $v_3$ qualifies as an instrument for $d$ when the edge for $b$ is removed (see Figure \ref{fig:IS-ECb}).\footnote{Note that if $|Z|=1$, then the conditions of Theorem \ref{thm:aux-IS} are satisfied if $Z$ is an instrumental set in $G_{E\cup E_{z}}$.}  Now, the identification of $d$ allows us to identify $a$ and $c$ using $v_5^*$ in $G^{d+}$, since $v_5$ is an instrument for $a$ and $c$ when the edge for $d$ is removed (see Figure \ref{fig:IS-ECc}).  

The above example also demonstrates that certain coefficients are identified only after using auxiliary instrumental sets iteratively.  We now define \emph{aux-IS identifiability}, which characterizes when a set of coefficients is identifiable using auxiliary instrumental sets.

\begin{mydef}[Aux-IS Identifiability]
\label{def:IS+EC}
Given a graph $G$, a set of directed edges $E$ is \emph{aux-IS identifiable} if there exists a sequence of sets of directed edges $(E_1 , E_2 , ... E_k)$ s.t.
\begin{enumerate}[(i)]
\item $E_1$ is identified using instrumental sets in $G$,
\item $E_i$ is identified using auxiliary instrumental sets for all $i\in \{2, 3, ... , k\}$ in $G^{E^{'}+}$ where $E^{'}\subseteq E_1 \cup E_2 \cup ... \cup E_{i-1}$,
\item and $E$ is identified using auxiliary instrumental sets in $G^{E_R+}$, where $E_R \subseteq (E_1 \cup E_2 \cup ... \cup E_k)$.
\end{enumerate}
\end{mydef}

\section{Auxiliary Instrumental Sets and the Half-Trek Criterion}
\label{sec:IS}

In this section, we explore the power of auxiliary instrumental sets, ultimately showing that they are at least as powerful as the g-HTC.  Having defined auxiliary instrumental sets, we now briefly describe the g-HTC. The g-HTC is a generalization of the half-trek criterion that allows the identification of arbitrary coefficients rather than the whole model \cite{chen:15}.\footnote{If any coefficient is not identified, then the half-trek criterion algorithm will simply output that the model is not identified.}  First, we give the definition for the general half-trek criterion, then we will discuss how it can be used to identify coefficients before showing that any g-HTC identifiable coefficient is also aux-IS identifiable.  

\begin{mydef}[General Half-Trek Criterion] Let $E$ be a set of directed edges sharing a single head $y$.  A set of variables $Z$ satisfies the \emph{general half-trek criterion} with respect to $E$, if 

\begin{enumerate}[(i)]
\item $|Z|=|E|$,
\item $Z\cap({y}\cup Sib(y))=\emptyset$, 
\item There is a system of half-treks with no sided intersection from $Z$ to $Ta(E)$, and
\item $(Pa(y)\setminus Ta(E))\cap htr(Z) = \emptyset$.
\end{enumerate}
\end{mydef}

A set of directed edges, $E$, sharing a head $y$ is identifiable if there exists a set, $Z_E$, that satisfies the general half-trek criterion (g-HTC) with respect to $E$, and $Z_E$ consists only of ``allowed" nodes.  Intuitively, a node $z$ is allowed if $E_{zy}$ is identified or empty, where $E_{zy}$ is the set of edges that \begin{enumerate}[(i)]
\item lie on half-treks from $y$ to $z$ or
\item lie on paths between $z$ and $Pa(y) \setminus Ta(E)$.
\end{enumerate}

We will continue to use the $E_{Zy}$ notation and allow $Z$ to be a set of nodes.  
When $Z$ satisfies the g-HTC and consists only of allowed nodes for $E$, we say that $Z$ is a \emph{g-HT admissible set} for $E$.  If a g-HT admissible set exists for $E$, then the coefficients in $E$ are \emph{g-HT identifiable}.  The lemma below characterizes when a set of parameters is g-HT identifiable.  This characterization parallels Definition \ref{def:IS+EC} and will prove useful in the proofs to follow.  

\begin{figure}
\centering
\begin{subfigure}[t]{.13\textwidth}
\caption{}
\label{fig:zIDa}
\includegraphics[width=\textwidth]{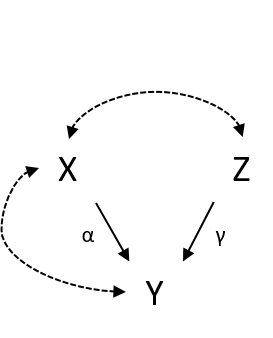}
\end{subfigure}
\begin{subfigure}[t]{.136\textwidth}	
\caption{}
\label{fig:zIDb}
\includegraphics[width=\textwidth]{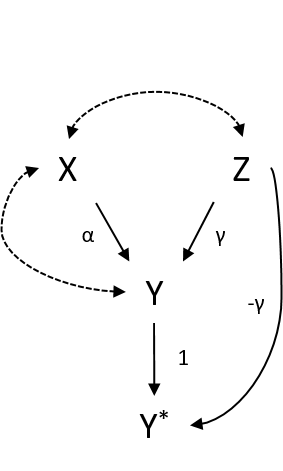}
\end{subfigure}
\begin{subfigure}[t]{.141\textwidth}
\caption{}
\label{fig:zIDc}
\includegraphics[width=\textwidth]{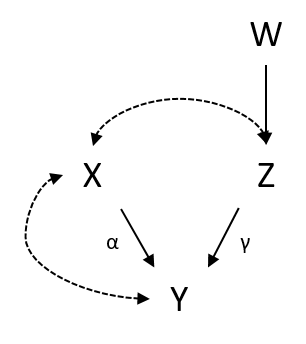}
\end{subfigure}
\caption{(a) The effect of $x$ on $y$ is not identifiable in non-parametric models, even with experiments over $z$ (b) The augmented graph, $G^{\gamma+}$, where $\alpha$ is identified using $z$ as a quasi-instrument, if we assume linearity (c) $\{x, w\}$ is an instrumental set for $\{\alpha,\gamma\}$.}
\end{figure}

\begin{lemma}
\label{lem:g-HT-ID}
If a set of directed edges, $E$, is g-HT identifiable, then there exists sequences of sets of nodes, $(Z_1 , ... Z_k)$, and sets of edges, $(E_1 , ... , E_k)$, such that 
\begin{enumerate}
\item $Z_i$ satisfies the g-HTC with respect to $E_i$ for all $i\in \{1, ... , k\}$,
\item $E_{Z_1 y_1}=\emptyset$, where $y_i = He(E_i)$ for all $i\in \{1, ... , k\}$, and
\item $E_{Z_i y_i}\subseteq (E_1 \cup ... \cup E_{i-1})$ for all $i\in \{1,... k\}$.
\end{enumerate}
\end{lemma}
\begin{proof}
The lemma follows from Theorem 1 in \cite{chen:15}. 
\end{proof}

To see how the g-HTC can be used to identify coefficients, consider again Figure \ref{fig:IS-ECa}.  Initially, only $b$ is identifiable.  We are able to use $\{v_2\}$ or $\{v_1\}$ as a g-HT admissible set for $b$ since $E_{v_2 v_3}$ and $E_{v_1 v_3}$ are equal to $\emptyset$.  All other nodes are half-trek reachable from $v_2$ and their edges on the half-trek from $v_2$ are not identified.  Once $b$ is identified, we can use $\{v_3\}$ as a g-HT admissible set to identify $d$.  Similarly, once $d$ is identified, we can use $v_5$ as a g-HT admissible set to identify $a$ and $c$. 

The following lemma connects g-HT-admissibility with auxiliary instrumental sets.  

\begin{lemma}
\label{lem:gHTC=>IS+R}
If $Z$ is a g-HT-admissible set for a set of directed edges $E$ with head $y$, then $E$ is identified using instrumental sets in $G^{E_{Zy}+}$. 
\end{lemma}

Now, we are ready to show that aux-IS identifiability subsumes g-HT identifiability.  

\begin{theorem}
\label{thm:gHTC=>IS+R}
Given a semi-Markovian linear SEM with graph $G$, if a set of edges, $E$, with head $y$, is g-HTC identifiable, then it is aux-IS identifiable.
\end{theorem}
\begin{proof}
Since $E$ is g-HTC identifiable, there exists sequences of sets of nodes, $(Z_1 , ... Z_k)$, and sets of edges, $(E_1 , ... , E_k)$, such that 
\begin{enumerate}
\item $Z_i$ satisfies the g-HTC with respect to $E_i$ for all $i\in \{1, ... , k\}$,
\item $E_{Z_1 y_1}=\emptyset$, where $y_i = He(E_i)$ for all $i\in \{1, ... , k\}$, and
\item $E_{Z_i y_i}\subseteq (E_1 \cup ... \cup E_{i-1})$ for all $i\in \{2, ... , k\}$.
\end{enumerate}

Now, using Lemma \ref{lem:gHTC=>IS+R}, we see that there $Z_1$ is an instrumental set for $E_1$ in $G^{E_{Z_1 y_1}+} = G$ and $E_i$ is identified using instrumental sets and Lemma 3 in $G^{E_{Z_i y_i}+}$ with $E_{Z_i y_i}\subseteq (E_1 \cup ... \cup E_{i-1})$ for all $i\in \{2, ... , k\}$.  As a result, $E$ is Aux-IS identifiable.
\end{proof}

\section{Further Applications}
\label{sec:discuss}

We have formalized auxiliary variables and demonstrated their ability to increase the identification power of instrumental sets.  In this section, we discuss additional applications of auxiliary variables as alluded to in the introduction, namely, incorporating external knowledge of coefficients values and deriving new constraints over the covariance matrix.

When the causal effect of $x$ on $y$ is not identifiable and performing randomized experiments on $x$ is not possible (due to cost, ethical, or other considerations), we may nevertheless be able to identify the causal effect of $x$ on $y$ using knowledge gained from experiments on another set of variables $Z$.  The task of determining whether causal effects can be computed using surrogate experiments generalizes the problem of identification and was 
named $z$-identification in \cite{bareinboim:pea12}. They provided necessary and sufficient conditions for this task in the non-parametric setting. Considering Figure \ref{fig:zIDa}, one can immediately see that the effect of $x$ on $y$ is not identifiable, given the unblockable back-door path. Additionally, using BP's z-identification condition, one can see that the effect of $x$ on $y$ is not identifiable, even with experiments over $z$. Remarkably, if one is willing to assume that the system is linear, more can be said.  The experiment over $z$ would yield $\gamma$, allowing us to create an auxiliary variable, $y^*$, which is represented by Figure \ref{fig:zIDb}. Now, $\alpha$ can be easily identified using auxiliary variables.  To witness, note that $\sigma(z,y^*)=C_{xz}\alpha +\gamma - \gamma$ and $\sigma(z,x)=C_{xz}$ so that $\alpha = \frac{\sigma(z,y^*)}{\sigma(z,x)}$.  


While $z$ is not technically an instrument for $\alpha$ in $G^{\gamma+}$, it behaves like one.  When $z$ allows the identification of $\alpha$ by using an auxiliary variable $y^*$, we will call $z$ a \emph{quasi-instrument}.  The question naturally arises whether we can improve aux-IS identifiability (Def. \ref{def:IS+EC}) by using quasi-instruments.  However, aux-IS identifiability requires that we learn the value of $\gamma$ from the model, not externally.  In order to identify $\gamma$ from the model, we would require an instrument.  If such an instrument, $w$, existed, as in Figure \ref{fig:zIDc}, then both $\gamma$ and $\alpha$ could have been identified together using $\{z,w\}$ as an instrumental set.  As a result, quasi-instruments are not necessary.  However, if $\gamma$ could only be evaluated externally, then quasi-instruments are necessary to identify $\alpha$.  
%
%

In some cases, the cancellation of paths due to auxiliary variables may generate new vanishing correlation constraints.  For example, in Figure \ref{fig:IS+Rb}, we have that $\sigma(z^*, s) = \beta \gamma -\beta \gamma = 0$.  Thus, we see that auxiliary variables allows us to identify additional testable implications of the model.  Moreover, if certain coefficients are evaluated externally, that information can also be used to generate testable implications.  Lemma \ref{lem:sep} can be used to identify independences involving auxiliary variables from the graph, $G$.

Besides z-identification and model testing, these new constraints can also be used to prune the space of compatible models in the task of structural learning. Additionally, it is natural to envision that auxiliary variables can be useful to answer evaluation questions in different, but somewhat related domains, such as in the transportability problem \cite{pearl:bar11-r372a}, or more broadly, the data-fusion problem \cite{bareinboim:pea15-r450}, where datasets collected under heterogenous conditions need to be combined to answer a query in a target domain.

\section{Conclusion}
\label{sec:conclusion}

In this paper, we tackle the fundamental problem of identification in linear system as articulated by \cite{fisher:66}. We move towards a general solution of the problem, enriching graph-based identification and model testing methods by introducing auxiliary variables.  Auxiliary variables allows existing identification and model testing methods to incorporate knowledge of non-zero parameter values.  We proved independence properties of auxiliary variables and demonstrated that by iteratively identifying parameters using auxiliary instrumental sets, we are able to greatly increase the power of instrumental sets, to the extent that it subsumes the most general criterion for identification of linear SEMs known to date.  We further discussed how auxiliary variables can be useful for the general tasks of testing and $z$-identification.  

\section{Acknowledgments}

This research was supported in parts by grants from NSF \#IIS-1302448 and \#IIS-1527490 and ONR \#N00014-13-1-0153 and \#N00014-13-1-0153.

\section{Appendix}

\setcounter{lemma}{0}
\setcounter{theorem}{0}

\begin{lemma}
\label{lem:aug}
Given a linear structural model, $M$, with induced graph $G$, and a set of directed edges $E$ with known coefficient values, $\sigma (w,v) = W_{G^{E+}}(w,v)$, where $w,v \in V^*$ and $w\neq v$ .
\end{lemma}
\begin{proof}
First, we consider the case where neither $w$ nor $v$ are in $He^* = \{y^*_1 , ... , y^*_k\}$.  In this case, it should be clear that $W_G (w,v) = W_{G^{E+}}(w,v)$.  Now, since $w$ and $v$ are the same variables in $M$ as they are in $M^{E+}$, we have that Wright's rules hold for $\sigma(x,y)$ in $M^{E+}$ using $G^{E+}$.  

Next, we consider that case when one of $w$ or $v$ equals $y^*_i$ for some $i\in \{1, ..., k\}$.  Without loss of generality, let us assume that $w=y^*_i$.  First note that $y^*_i = y_i - \Lambda_{X_i y_i} T^t_i$ so 
\begin{align*}
\sigma(y^*_i, v) &= \sigma(y_i-\lambda^1_i x^1_i -... - \lambda^l_i x^l_i , v) \\ 
&= 1*\sigma(y_i, v)- \lambda^1_i \sigma(x^1_i, v)- ... - \lambda^l_i \sigma(x^l_i, v),
\end{align*} 
where $\Lambda_{X_i y_i} = (\lambda^1_i , ... , \lambda^l_i )^T$ and $X_i = (x^1_i , ... , x^l_i)^T$. Additionally, note that the only edges which are connected to $y^*_i$ are the directed edges, $(x^1_i, y^*_i), ... , (x^l_i , y^*_i)$ with coefficients $\lambda^1_i , ... , \lambda^l_i $, and $(y_i, y^*_i)$ with the coefficient 1.  As a result, 
\begin{align*}
W_{G^{E+}} (y^*_i, v) =& 1 * W_{G^{E+}} (y_i, v) - \lambda^1_i W_{G^{E+}}(x^1_i,v) - \\ & ... - \lambda^l_i W_{G^{E+}}(x^l_i,v) \\ 
=& 1*\sigma(y_i, v)- \lambda^1_i \sigma(x^1_i, v)- ... \\ &- \lambda^l_i \sigma(x^l_i, v) \\=& \sigma(y^*_i, v).
\end{align*}

Finally, we consider the where $w=y^*_i$ and $v=y^*_j$ for some $i,j \in \{1,...,k\}$.  Here, \begin{align*}\sigma(y^*_i, y^*_j)=&\sigma(y_i-\lambda^1_i x^1_i -... - \lambda^l_i x^l_i , \\
&\;\;\:\: y_j-\lambda^1_j x^1_j -... - \lambda^m_j x^m_j)\\
=&(1*\sigma(y_i , y_j)*1- 1 *\sigma(y_i, x^1_j) *\lambda^1_j \\
 &- ...-1*\sigma(y_i , x^m_j)*\lambda^m_j)+\\
 &((-\lambda^1_i *\sigma(x^1_i , y_j)*1)+\lambda^1_i*\sigma(x^1_i , x^1_j)*\lambda^1_j\\
 &+...+\lambda^1_i*\sigma(x^1_i , x^m_j)*\lambda^m_j)+\\
& \vdots \\
&+(-\lambda^l_i * \sigma(x^l_i , y_j)*1)+ \lambda^l_i *\sigma(x^l_i , x^1_j) \lambda^1_j \\
&+... +\lambda^l_i *\sigma(x^l_i , x^m_j)*\lambda^m_j)).
\end{align*}
Similarly, 
\begin{align*}
W_{G^{E+}}(y^*_i , y^*_j) =&(1*W_{G^{E+}}(y_i , y_j)*1 - 
\\&1* W_{G^{E+}}(y_i, x^1_j) *\lambda^1_j - ... -\\
&1*W_{G^{E+}}(y_i , x^m_j)*\lambda^m_j) + \\
&((-\lambda^1_i *W_{G^{E+}}(x^1_i , y_j)*1)+\\
&\lambda^1_i*W_{G^{E+}}(x^1_i , x^1_j)*\lambda^1_j+...+\\
&\lambda^1_i*W_{G^{E+}}(x^1_i , x^m_j)*\lambda^m_j)\\
&\vdots \\
&+(-\lambda^l_i * W_{G^{E+}}(x^l_i , y_j)*1)+ \\
&\lambda^l_i *W_{G^{E+}}(x^l_i , x^1_j) \lambda^1_j +... \\
&+\lambda^l_i *W_{G^{E+}}(x^l_i , x^m_j)*\lambda^m_j))\\
&=\sigma(y^*_i , y^*_j). 
\end{align*}
\end{proof}

\begin{lemma}
\label{lem:sep}
Given a semi-Markovian linear SEM with graph $G$, $(z^*\indep y)_{G^{E_{z}+}}$ if and only if $z$ is d-separated from $y$ in $G_{E_z}$, where $E_z\subseteq Inc(z)$ and $G_{E_z}$ is the graph obtained when $E_z$ is removed from $G$. 
\end{lemma}
\begin{proof}
First, we show the sufficiency of Lemma \ref{lem:sep}.  Suppose that $(z^*\notindep y)_{G^{E_{z}+}}$.  Then $W_{G^{E_{z}+}}(x,y)\neq 0$.  First, since $(z\indep y)_{G_{E_z}}$, there are no unblocked paths between $z$ and $y$ that do not include the edge $z\leftarrow w$ for some $w\in Pa(z)$.  As a result, there are no blocked paths between $z^*$ and $y$ that do not either include the edge $z^* \leftarrow z \leftarrow w$ or $z^* \leftarrow w$.  Now, Lemma \ref{lem:cancel} tells us that for each path, $\pi_{e}$, beginning $z^* \leftarrow z \leftarrow w$, there is a corresponding path beginning $z^*\leftarrow w$ for which $P(\pi_e)+P(\pi_{-e})=0$.  As a result, for $W_{G^{E_{z}+}_E}(x,y)\neq 0$ there must be an unblocked path, $\pi$ from $z^*$ to $y$ beginning $z^* \leftarrow w ... y$.  Lemma \ref{lem:excess} tells us that such paths must include $z$ as an internal node.  Since $(z\indep y)_{G}$, and there are no unblocked paths between $z$ and $y$ that do not include the edge $z\leftarrow w$, $\pi$ must be of the form $z^* \leftarrow w ... z \leftarrow w ...y$.  However, this path visits $w$ twice and is therefore not an unblocked path.  As a result, $W_{G^{E_{z}+}_E}(x,y)=0$, and we have a contradiction.

Now, we show the necessity of Lemma \ref{lem:sep}.  If $z$ is not d-separated from $y$ in $G_{E_z}$, then there exists some path, $\pi$, between $z$ and $y$ that does not include an edge in $E_z$.  As a result, there is a path in $G^{E_z+}$, from $z^*$ to $y$, which is the concatenation of $z^* \leftarrow z$ with $\pi$.  Clearly, this path is not cancelled out in the way paths that go through $E_z$ are.  As a result, we have that $\sigma(z^*,y)\neq 0$.  
\end{proof}

\begin{theorem}
\label{thm:aux-IS}
Let $E_Z$ be a set of directed edges whose coefficient values are known.  A set of directed edges, $E=\{(x_1, y) , ..., (x_k, y)\}$, in a graph, $G$, is identified if there exists $Z$ such that:
\begin{enumerate}
\item $|Z|=k$,
\item for all $z_i \in Z$, $(z_i \indep y)_{G_{E\cup E_{z_i}}}$, where $E_{z_i}=E_Z\cap Inc(z_i)$ and $G_{E\cup E_{z_i}}$ is the graph obtained by removing the edges in $E\cup E_{z_i}$ from $G$, and
\item there exists unblocked paths $\Pi = \{\pi_1 , \pi_2 ,... , \pi_k\}$ such that $\pi_i$ is a path from $z_i$ to $x_i$ and $\Pi$ has no sided intersection.
\end{enumerate}
\end{theorem}
If the above conditions are satisfied then $Z$ is an auxiliary instrumental set for $E$.  
\begin{proof}
First, note that if $E_{Z} = \emptyset$, then $Z$ is an instrumental set in $G$ and we are done.  We now consider the case when $E_{Z} \neq \emptyset$.  Since $|Z^*|=|Z|-|A|+|A*| = |Z|-|A|+|A|=|Z|$, $|Z^*|=|E|$, IS-(i) is satisfied.  Now, we show that IS-(iii) is satisfied.  For each $z_i \in Z$, let $\pi_{z_i} \in \Pi$ be the path in $\Pi$ from $z_i$ to $Ta(E)$.  Now, for each $a^*_i\in A^*$, let $\pi_{a^*_i}$ be the concatenation of path $a^* \leftarrow a$ with $\pi_{a_i}$.  It should be clear that $\Pi\setminus \{\pi_{a_i}\} \cup \{\pi_{a^*_i}\}$ satisfies IS-(iii) in $G^{E_{Zy}+}$.  Lastly, we need to show that IS-(ii) is also satisfied.   

First, if $z_i \in Z\setminus A$, then $(z_i \indep y)_{G_E}$.  It follows follows that $(z_i \indep y)_{G^{E_Z+}_E}$ since Lemma \ref{lem:nointernalZ*} implies that no new paths from $z_i$ to $y$ can be generated by adding the auxiliary nodes.  Now, $(a_i\indep y)_{G^{E_Z+}_E}$ for all $a_i \in A$ follows from (ii) and Lemma \ref{lem:sep}, proving the theorem.
\end{proof}

\stepcounter{lemma}

\begin{lemma}
\label{lem:gHTC=>IS+R}
If $Z$ is a g-HT-admissible set for a set of directed edges $E$ with head $y$, then $E$ is identified using instrumental sets in $G^{E_{Zy}+}$. 
\end{lemma}
\begin{proof}
First, note that if $E_{Zy} = \emptyset$, then $Z$ is an instrumental set in $G$ and we are done.  We now consider the case when $E_{Zy} \neq \emptyset$.  

Let $A = He(E_{Zy})$.  We will show that $Z^* = (Z\setminus A)\cup A^*$ is an instrumental set in $G^{E_{Zy}+}$.  From HT-(i), we have that $|Z|=|E|$.  Since $|Z^*|=|Z|-|A|+|A*| = |Z|-|A|+|A|=|Z|$, $|Z^*|=|E|$ and IS-(i) is satisfied.   From HT-(iii), we have that $\Pi$ is a set of paths from $Z$ to $Ta(E)$ with no sided intersection.  For each $z_i \in Z$, let $\pi_{z_i} \in \Pi$ be the path in $\Pi$ from $z_i$ to $Ta(E)$.  Now, for each $a^*_i\in A^*$, let $\pi_{a^*_i}$ be the concatenation of path $a^* \leftarrow a$ with $\pi_{a_i}$.  It should be clear that $\Pi\setminus \{\pi_{a_i}\} \cup \{\pi_{a^*_i}\}$ satisfies IS-(iii) in $G^{E_{Zy}+}$.  We need to show that IS-(ii) is also satisfied.   

Consider any $z\in Z$ such that $(z\notindep y)_{G_E}$.  In order for $Z$ to be a g-HT-admissible set, any path, $\pi_z$ connecting $z$ with $y$ in $G_E$ must include an edge, $e_z$ in $Inc(z)$.  Moreover, $e_z \in E_{Zy}$.  As a result, $(z\indep y)_{G_{E\cup E_{Zy}}}$.  It follows from Lemma \ref{lem:sep} that $(z^* \indep y)_{G^{E_{Zy}+}_E}$.  As a result, $Z^*$ is an instrumental set for $E$ in $G^{E_{Zy}+}$.  
\end{proof}

\begin{theorem}
\label{thm:gHTC=>IS+R}
Given a semi-Markovian linear SEM with graph $G$, if a set of edges, $E$, with head $y$, is g-HTC identifiable, then it is IS+EC identifiable.
\end{theorem}
\begin{proof}
Since $E$ is g-HTC identifiable, there exists sequences of sets of nodes, $(Z_1 , ... Z_k)$, and sets of edges, $(E_1 , ... , E_k)$, such that 
\begin{enumerate}
\item $Z_i$ satisfies the g-HTC with respect to $E_i$ for all $i\in \{1, ... , k\}$,
\item $E_{Z_1 y_1}=\emptyset$, where $y_i = He(E_i)$ for all $i\in \{1, ... , k\}$, and
\item $E_{Z_i y_i}\subseteq (E_1 \cup ... \cup E_{i-1})$ for all $i\in \{2, ... , k\}$.
\end{enumerate}

Now, using Lemma \ref{lem:gHTC=>IS+R}, we see that there $Z_1$ is an instrumental set for $E_1$ in $G^{E_{Z_1 y_1}+} = G$ and $E_i$ is identified using instrumental sets and Lemma 3 in $G^{E_{Z_i y_i}+}$ with $E_{Z_i y_i}\subseteq (E_1 \cup ... \cup E_{i-1})$ for all $i\in \{2, ... , k\}$.  As a result, $E$ is Aux-IS identifiable.
\end{proof}

\stepcounter{lemma}

The following lemmas characterize the cancellation of back-door paths that occurs as a result of adding auxiliary variables, ensure that no new problematic paths are created, and are used to prove Lemma \ref{lem:sep}.

\begin{lemma}
\label{lem:cancel}
Let $E$ be a set of directed edges whose coefficient values are known.  For any $e\in E$, let $z_e = He(e)$ and $w_e=Ta(e)$. If $\pi_e$ is an unblocked path in $G^{E+}$ that is the concatenation of $z^* \leftarrow z\leftarrow w$ and some path $\pi_w$ beginning from $w$, then there exists $\pi_{-e}$ such that $\pi_{-e}$ is the concatenation of $z^* \leftarrow w$ and $\pi_w$ and $P(\pi_e)+P(\pi_{-e})=0$.
\end{lemma}
\begin{proof}
For each $e\in E$ with $c_e$ the coefficient value for $e$, $z_e = He(e)$, and $w_e = Ta(e)$, the graph $G^{E+}$ adds a node $z^*$, a directed edge from $z$ to $z^*$ with coefficient value 1, and a directed edge from $w$ to $z^*$ with coefficient value $-c_e$.  As a result, we can construct $\pi_{-e}$ by concatenating $z_e^* \leftarrow w$ with $\pi_w$.  $P(\pi_e)+P(\pi_{-e})=0$ follows from the fact that the edge $(z, z^*)$ has a coefficient value of 1 and $(w,z^*)$ has a coefficient value of $-c_e$.
\end{proof}

The following lemma characterizes the ``new'' back-door paths that are generated as a result of adding an auxiliary variable.  For example, in Figure 1b, we have a ``new'' back-door path between $z^*$ and $y$, $z^*\leftarrow w \leftarrow z \rightarrow x \rightarrow y$.  We will see in Theorem \ref{thm:aux-IS} that such paths are not problematic for identification using instrumental sets in the augmented model since they must include $z$.  

\begin{lemma}
\label{lem:excess}
Let $E$ be a set of directed edges whose coefficient values are known.  For any $e\in E$, let $z_e = He(e)$, $w_e=Ta(e)$, and $\pi_{-e}$ be a path in $G^{E+}$ that is the concatenation of $z^* \leftarrow w$ and some path $\pi_w$ beginning from $w$.  If there does not exist an unblocked path, $\pi_e = z^* \leftarrow z \leftarrow w + \pi_w$, so that $P(\pi_e)+P(\pi_{-e})=0$, then $\pi_{-e}$ is of the form $z^*\leftarrow w ... z ...$
\end{lemma}
\begin{proof}
By construction, the only way $\pi_{-e}$ does not have a corresponding unblocked path $\pi_e$ is if $\pi_{-e}$ goes through $w$ so that $\pi_e$ would have the form $z^* \leftarrow z \leftarrow w ... z ...$, which is not an unblocked path because it visits $z$ twice.
\end{proof}

The next lemma ensures that no new paths between the nodes in $G$ are created in $G^{E+}$.  

\begin{lemma}
\label{lem:nointernalZ*}
Let $V$ be the set of variables in a semi-Markovian linear SEM, $M$, with graph $G$.  Let $E$ be a set of directed edges whose coefficient values are known and let $Z=He(E)$.  If $\pi$ is an unblocked path in $G^{E+}$ between $x$ and $y$, with $x,y\in V$,  then $\pi$ is an unblocked path between $x$ and $y$ in $G$.  
\end{lemma}
\begin{proof}
Let $Z=He(E)$.  The only edges added to $G^{E+}$ from $G$ are connected are pointed at $Z^*$.  As a result, if a member of $Z^*$ was an internal node on $\pi$, then it would be a collider, and the path would be blocked.  As a result, $\pi$ contains no members of $Z^*$ and must also be an unblocked path in $G$.  
\end{proof}

\bibliography{book}  
\bibliographystyle{named}
\end{document}